\documentclass{article}

\usepackage{epsfig}
\usepackage{amssymb}
\usepackage{latexsym}
\usepackage{amsthm}
\usepackage{amsmath}
\usepackage{multicol}
\newtheorem{theorem}[section]{Theorem}
\newtheorem{corollary}[section]{Corollary}

\begin{document}

\date{written 2011}

\title{New tight approximations for Fisher's exact test}

\author{Wilhelmiina H{\"a}m{\"a}l{\"a}inen\\
School of Computing\\
University of Eastern Finland\\ 
Finland\\
whamalai@cs.uef.fi}

\maketitle

\begin{abstract}
Fisher's exact test is often a preferred method to estimate the
significance of statistical dependence. However, in large data sets
the test is usually too worksome to be applied, especially in an
exhaustive search (data mining). The traditional solution is to
approximate the significance with the $\chi^2$-measure, but the
accuracy is often unacceptable. As a solution, we introduce a family
of upper bounds, which are fast to calculate and approximate Fisher's
$p$-value accurately. In addition, the new approximations are not
sensitive to the data size, distribution, or smallest expected counts
like the $\chi^2$-based approximation. According to both theoretical
and experimental analysis, the new approximations produce accurate
results for all sufficiently strong dependencies. The basic form of
the approximation can fail with weak dependencies, but the general
form of the upper bounds can be adjusted to be arbitrarily accurate.
\end{abstract}


\section{Introduction}

Pattern recognition algorithms often involve testing the significance
of a statistical dependency between two binary variables $X$ and $A$,
given the observed counts $m(XA)$, $m(X\neg A)$, $m(\neg XA$, and
$m(\neg X\neg A)$. If the test is done only a couple of times, the
computation time is not crucial, but in an exhaustive search the test
may repeated thousands or even millions of times. 

For example, in data mining a classical problem is to search for the
most significant classification rules of the form $X\rightarrow A$ or
$X\rightarrow \neg A$, where $X$ is a set of binary attributes and $A$
is a binary class attribute. This problem is known to be $NP$-hard
with common significance measures like the $\chi^2$-measure
\cite{morishitasese} and no polynomial time solutions are known. Even
a more complex problem is to search for all sufficiently significant
dependency rules, where the consequence attribute is not fixed (see
e.g.\ \cite{whthesis}). In both problems the number of all tested
patterns can be exponential and therefore each rule should be tested
as fast as possible, preferrably in a constant time.

The problem is that typically the mined data sets are very large (the
number of attributes can be tens of thousands and the number of rows
millions). Still, the most significant (non-trivial) dependency rules
may be relatively infrequent and the corresponding distributions too
skewed for fast but inaccurate asymptotic tests. For accurate results,
one should test the significance of dependency $X\rightarrow A$ with
Fisher's exact test, which evaluates the exact probability of the
observed or a stronger dependency in the given data, if $X$ and $A$
were actually independent. For a positive dependency between $X$ and
$A$ the probability ($p$-value) is defined by the cumulative
hypergeometric distribution 

$$p_F(X\rightarrow A)=\sum_{i=0}^J\frac{\left(m(X) \atop m(XA)+i \right) \left(m(\neg X) \atop m(\neg X\neg A)+i \right)}{\left(n \atop m(A)\right)},$$

where $m(Z)$ is the absolute frequency of set $Z$, $n$ is the data size
(number of rows), and $J=\min\{m(X\neg A),m(\neg XA)\}$. (For a
negative dependency between $X$ and $A$ it suffices to replace $A$ by
$\neg A$.)

However, Fisher's exact test is computationally demanding, when the
data size $n$ is large. In each test one should evaluate $J+1\leq
\frac{n}{4}+1$ terms, which means that the worst case time complexity
is ${\mathcal{O}}(n)$. For example, if $m(X)=m(A)=500 000$ and
$m(XA)=300 000$, we should evaluate 200 001 terms. In addition, each
term involves binomial factors, but they can be evaluated in a
constant time, if all factorials $i!$, $i=1,\hdots,n$ have been
tabulated in the beginning. 

A common solution is to estimate the $p$-values from the
$\chi^2$-measure, when the data size is moderate or large. The
justification is that asymptotically (when $n$ approaches $\infty$)
$p_F$ can be approximated by the $\chi^2$-based $p$-values. However,
in finite data sets the approximations are often quite inaccurate. The
reason is that the $\chi^2$-measure is very sensitive to the data
distribution. If the exact hypergeometric distribution is symmetric,
the $\chi^2$-measure works well, but the more skewed the distribution
is, the more inaccurate the approximated $p$-values are
\cite{yates,agresti}. Classically \cite{fisher1925}, it is recommended
that the $\chi^2$-approximation should not be used, if any of the
expected counts (under the assumption of independence between $X$ and
$A$) $nP(X)P(A)$, $nP(X)P(\neg A)$, $nP(\neg X)P(A)$, or $nP(\neg
X)P(\neg A)$ is less than 5. However, this rule of thumb prevents only
the most extreme cases, which would lead to false discoveries. If the
problem is to search for e.g.\ the best 100 dependency rules, the
$\chi^2$-measure produces quite different results than $p_F$, even if
all expected counts and the data size are large.

In this paper, we introduce better approximations for the exact
$p_F$-values, which still can be calculated in a constant time. The
approximations are upper bounds for the $p_F$ (when the $\chi^2$-based
values are usually lower bounds, i.e.\ better than true values), but
when the dependency is sufficiently strong, they give tight
approximates to the exact values. In practice, they give identical
results with the exact $p_F$-values, when used for rule ranking.  

The idea of the approximations is to calculate only the first or a
couple of first terms from $p_F$ exactly and estimate an upper bound
for the rest terms. The simplest upper bound evaluates only the first
term exactly. It is also intuitively the most appealing as a goodness
measure, because it is reminiscent to the existing dependency measures
like the odds ratio. When the dependencies are sufficiently strong (a
typical data mining application), the results are also highly
accurate. However, if the data set contains only weak and relatively
insignificant dependencies, the simplest upper bound may produce too
inaccurate results. In this case, we can use the tighter upper bounds,
one of which can be adjusted to arbitrary accurate. However, the more
accurate $p$-values we want to get, the more terms we have to
calculate exactly. Fortunately, the largest terms of $p_F$ are always
the first ones, and in practice it is sufficient to calculate only a
small fraction (typically 2--10) of them exactly.

The rest of the paper is organized as follows. In Section
\ref{sec2} we introduce the upper bounds and give error bounds for
the approximations. In Section \ref{seceval} we evaluate the upper
bounds experimentally, concentrating on the weak (and potentially the
most problematic) dependencies. The final conclusions are drawn in
Section \ref{secconcl}.

\begin{table}
\caption{Basic notations.}
\label{notations}
\begin{center}
\small{
\begin{tabular}{ll}
$A$&a single binary attribute\\
$A\equiv(A=1)$&a short-hand notation for an event,\\
&where $A$ is true\\
$\neg A\equiv(A=0)$&a short-hand notation for an event,\\
&where $A$ is false\\
$X=\{A_1,\hdots,A_l\}$&a set of binary attributes\\
$X\equiv (A_1=1\wedge \hdots A_l=1)$&a short-hand notation for an event,\\
&where all attributes $A_i\in X$, $|X|=l$,\\
&are true\\
$\neg X\equiv \neg(A_1=1 \wedge \hdots A_l=1)$&a short-hand notation for an event,\\
&where $A_i=0$ for some $A_i\in X$\\
$n$&data size (number of rows)\\
$m(X)$&absolute frequency of $X$\\
$P(X)=\frac{m(X)}{n}$&relative frequency of $X$\\
$\delta(X,A)=P(XA)-P(X)P(A)$&leverage\\
$\gamma(X,A)=\frac{P(XA)}{P(X)P(A)}$&lift\\
\end{tabular}
}
\end{center}
\end{table}

\section{Upper bounds}
\label{sec2}

The following theorem gives two useful upper bounds, which can be used
to approximate Fisher's $p_F$. The first upper bound is more accurate,
but it contains an exponent, which makes it more difficult to
evaluate. The latter upper bound is always easy to evaluate and also
intuitively appealing. 

\begin{theorem}
\label{ubtheorem}
Let us notate $p_F=p_0+p_1+\hdots +p_J$ and $q_i=\frac{p_i}{p_{i-1}}$, $i\geq 1$.
For positive dependency rule $X\rightarrow A$ with lift $\gamma(X,A)=\frac{p(XA)}{P(X)P(A)}$
\begin{align*}
p_F&\leq p_0\left(\frac{1-q_1^{J+1}}{1-q_1}\right)\\
&\leq p_0\left(1+\frac{1-P(A)\gamma(X,A)-P(X)\gamma(X,A)+P(X)P(A)\gamma(X,A)^2}{\gamma(X,A)-1}\right).
\end{align*}
\end{theorem}

\begin{proof}
Each $p_i$ can be expressed as $p_i=p_{abs}t_i$, where
$p_{abs}=\frac{m(A)!m(\neg A)!}{n!}$ is constant. Therefore, it is
enough to show the result for $p_X=\frac{p_F}{p_{abs}}$.

$p_X=t_0+t_1+...+t_J=t_0+q_1t_0+q_1q_2t_0+...+q_1q_2...q_Jt_0$, where

{\footnotesize
\begin{align*}
q_i&=\frac{t_i}{t_{i-1}}\\
&=\frac{(m(XA)+i-1)!(m(\neg X \neg A)+i-1)!(m(X \neg A)-i+1)! (m(\neg XA)-i+1)!}
{(m(XA)+i)!(m(\neg X \neg A)+i)!(m(X \neg A)-i)! (m(\neg XA)-i)!}=\\
&=\frac{(m(X \neg A)-i+1) (m(\neg XA)-i+1)}{(m(XA)+i)(m(\neg X \neg A)+i)}.
\end{align*}
}
Since $q_i$ decreases  when $i$ increases, the largest value is
$q_1$. We get an upper bound

$p_X=t_0+q_1t_0+q_1q_2t_0+...+q_1q_2...q_Jt_0 \leq t_0(1+q_1+q_1^2+q_1^3+...+q_1^J).$

The sum of geometric series is $t_0\frac{1-q_1^{J+1}}{1-q_1}$, which is 
the first upper bound. On the other hand, 
$$t_0\frac{1-q_1^{J+1}}{1-q_1}\leq\frac{t_0}{1-q_1}=t_0(1+\frac{q_1}{1-q_1}).$$
Let us insert $q_1=\frac{m(X \neg A)m(\neg X A)}{(m(XA)+1)(m(\neg X
  \neg A)+1)}$, and express the frequencies using lift
$\gamma=\gamma(X,A)$. For simplicity, we use notations $x=P(X)$ and
$a=P(A)$. Now $m(XA)=nxa\gamma$, $m(X\neg A)=nx-nxa\gamma$, $m(\neg
XA)=na-nxa\gamma$ and $m(\neg X\neg A)=n(1-x-a+xa\gamma)$. We get

\begin{align*}
\frac{q_1}{1-q_1}&=\frac{m(X\neg A)m(\neg XA)}{(m(XA)+1)(m(\neg X \neg
  A)+1)-m(X\neg A)m(\neg XA)}\\
&=\frac{n^2xa-n^2xa^2\gamma-n^2x^2a\gamma+n^2x^2a^2\gamma^2}{n^2xa\gamma+2nxa\gamma+n-nx-na+1-n^2xa}\\
&=\frac{nxa-nxa^2\gamma-nx^2a\gamma+nx^2a^2\gamma^2}{nxa\gamma+2xa\gamma+1-x-a+1/n-nxa}\\
\end{align*}

The nominator is $\geq nxa\gamma-nxa$, because
$2xa\gamma+1-x-a+1/n\geq 0 \Leftrightarrow$ $P(XA)+P(\neg X \neg
A)+1/n\geq 0$. Therefore
$$\frac{q_1}{1-q_1}\leq \frac{1-a\gamma -x\gamma+xa\gamma^2}{\gamma-1}.$$
\end{proof}

In the following, we will denote the looser (simpler) upper bound by
$ub1$ and the tighter upper bound (sum of the geometric series) by
$ub2$. In $ub1$, the first term of $p_F$ is always exact and the rest
are approximated, while in $ub2$, the first two terms are always exact
and the rest are approximated.

We note that $ub1$ can be expressed equivalently as 
$$ub1=p_0\left(\frac{P(XA)P(\neg X\neg A)}{\delta(X,A)}\right)=p_0\left(1+\frac{P(X\neg A)P(\neg XA)}{\delta(X,A)}\right),$$
where $\delta(X,A)=P(XA)-P(X)P(A)$ is the leverage. This is expression is closely related to the odds ratio
$$odds(X,A)=\frac{P(XA)P(\neg X\neg A)}{P(X\neg A)P(\neg XA)}=1+\frac{\delta}{P(X\neg A)P(\neg XA)},$$ 
which is often used to measure the strength of the dependency. The odds ratio can be expressed equivalently as
$$odds(X,A)=\frac{ub1}{ub1-1}.$$
We see that when the odds ratio increases (dependency becomes
stronger), the upper bound decreases. In practice, it gives a tight
approximation to Fisher's $p_F$, when the dependency is sufficiently
strong. The error is difficult to bind tightly, but the following
theorem gives a loose upper bound for the error, when $ub2$ is used for 
approximation.

\begin{theorem}
\label{ub2bound}
When $p_F$ is approximated by $ub2$, the error is bounded by 
$$err\leq p_0\left(\frac{q_1^2}{1-q_1}\right).$$
\end{theorem}

\begin{proof}
Upper bound $ub2$ can cause error only, if $J>1$. If $J=0$, $ub2=p_0$
and if $J=1$,
$ub2=p_0\left(\frac{1-q_1^2}{1-q_1}\right)=p_0(1+q_1)=p_0+p_1=p_F$.
Let us now assume that $J>1$. The error is
$err=ub2-p_F=p_0(1+q_1+\hdots+q_1^J-1-q_1-q_1q_2-\hdots -q_1q_2\hdots q_J)=p_0(q_1^2+q_1^3+\hdots q_1^J-q_1q_2-\hdots-q_1q_2\hdots q_J)$.
It has an upper bound
$$err<p_0q_1^2(1+q_1+\hdots +q_1^{J-2})=p_0q_1^2\left(\frac{1-q_1^{J-1}}{1-q_1}\right)\leq p_0\left(\frac{q_1^2}{1-q_1}\right).$$
\end{proof}

This leads to the following corollary, which gives good guarantees for
the safe use of $ub2$:

\begin{corollary}
If $\gamma(X,A)\geq \frac{1+\sqrt{5}}{2}\approx 1.62$, then $err=ub2-p_F\leq p_0$.
\end{corollary}

\begin{proof}
According to Theorem \ref{ub2bound}, $err\leq p_0$, if 
$q_1^2\leq 1-q_1$. This is true, when $q_1\leq \frac{\sqrt{5}-1}{2}$.

On the other hand, $q_1<\frac{1}{\gamma}$, when $\gamma>1$, because
\begin{multline*}
q_1=\frac{m(X\neg A)m(\neg XA)}{(m(XA)+1)(m(\neg X\neg A)+1)}<\frac{P(X\neg A)P(\neg XA)}{P(XA)P(\neg X\neg A)}=\frac{xa(1-a\gamma)(1-x\gamma)}{xa\gamma(1-x-a+xa\gamma)}.
\end{multline*}

This is $\leq \frac{1}{\gamma}$, because $1-a\gamma-x\gamma+xa\gamma^2\leq 1-x-a+xa\gamma \Leftrightarrow 0\leq x(\gamma-1)+a(\gamma-1)-xa\gamma(\gamma-1)=(x+a-P(XA))(\gamma-1)$.

A sufficient condition for $q_1\leq \frac{\sqrt{5}-1}{2}$ is that 
$$\frac{1}{\gamma}\leq \frac{\sqrt{5}-1}{2} \Leftrightarrow \gamma\geq \frac{2}{\sqrt{5}-1}=\frac{\sqrt{5}+1}{2}.$$

\end{proof}

This result also means that $ub2\leq 2p_F$, when the lift is as large
as required.

The simpler upper bound, $ub1$, can cause a somewhat larger error than
$ub2$, but it is even harder to analyze. However, we note that
$ub1=p_F$ only, when $J=0$. When $J=1$, there is already some error,
but in practice the difference is marginal. The following theorem
gives guarantees for the accuracy of $ub1$, when $\gamma\geq 2$.

\begin{theorem}
If $p_F$ is approximated with $ub1$ and $\gamma(X,A)\geq 2$, the error is 
bounded by $err\leq p_0$.
\end{theorem}

\begin{proof}
The error is $err=ub1-p_F=ub1-ub2+ub2-p_F$, where $ub2-p_F\leq
p_0\left(\frac{q_1^2}{1-q_1}\right)$ by Theorem \ref{ub2bound}. 

When $\gamma\geq 2$, $ub1$ (being a decreasing function of $\gamma$) is
$$ub1=p_0\left(\frac{\gamma(1-a\gamma-x\gamma+xa\gamma^2)}{\gamma-1}\right)\leq
p_02(1-a-x+2xa).$$ Therefore, the error is bounded by

\begin{align*}
err&\leq p_0\left(2(1-a-x+2xa)-\frac{1-q_1^{J+1}}{1-q_1}+\frac{q_1^2}{1-q}\right)\\
&=p_0\left(2-2a-2x+4xa -\frac{(1-q_1^2-q_1^{J+1})}{(1-q_1)}\right)\\
&=p_0\left(\frac{2(1-q_1)-2a(1-q_1)-2x(1-q_1)+4xa(1-q_1)-1+q_1^2+q_1^{J+1}}{1-q_1}\right)\\
&=p_0\left(\frac{1-2q_1+q_1^2+q_1^{J+1}2a(1-q_1)-2x(1-q_1)+4xa(1-q_1)}{1-q_1}\right).\\
\end{align*}

When $\gamma\geq 2$, $q_1\leq \frac{1}{2}$, and thus 
\begin{align*}
1-2q_1+q_1^2+q_1^{J+1}&\leq 1-q_1 \Leftrightarrow\\ 
0&\leq q_1-q_1^2-q_1^{J+1}\Leftrightarrow\\ 
0&\leq q_1(1-q_1-q_1^J).\\
\end{align*}
Therefore,
$$err\leq p_0\left(\frac{(1-q_1)(1-2a-2x+4xa)}{1-q_1}\right)=p_0(1-2a-2x+4xa).$$

The latter factor is always $\leq 1$, because
$2a+2x-4xa=2a(1-x)+2x(1-a)\geq 0$. Therefore $err\leq p_0$.
\end{proof}

Our experimental results support the theoretical analysis, according to
which both upper bounds, $ub1$ and $ub2$, give tight approximations to
Fisher's $p_F$, when the dependency is sufficiently strong. However,
if the dependency is weak, we may need a more accurate
approximation. A simple solution is to include more larger terms
$p_0+p_1+\hdots p_{l-1}$ to the approximation and estimate an upper
bound only for the smallest terms $p_{l}+\hdots p_J$ using the sum of
the geometric series. The resulting approximation and the
corresponding error bound are given in the following theorem. We omit
the proofs, because they are essentially identical with the previous
proofs for Theorems \ref{ubtheorem} and \ref{ub2bound}.

\begin{theorem}
\label{generalub}
For positive dependency rule $X\rightarrow A$ holds
$$p_F\leq p_0+\hdots p_{l-1}+p_{l}\left(\frac{1-q_{l+1}^{J-l+1}}{1-q_{l+1}}\right),$$
where
$$q_{l+1}=\frac{(m(X\neg A)-l)(m(\neg XA)-l)}{(m(XA)+l+1)(m(\neg X\neg A)+l+1)}$$ and $l+1\leq J$.

The error of the approximation is 
$$err\leq p_0\left(\frac{q_{l+1}^2}{1-q_{l+1}}\right).$$
\end{theorem}

\section{Evaluation}
\label{seceval}

Figure \ref{examplefig} shows the typical behaviour of the new upper
bounds, when the strength of the dependency increases (i.e.\ $m(XA)$
increases and $m(X)$ and $m(A)$ remain unchanged). In addition to upper
bounds $ub1$ and $ub2$, we consider a third upper bound, $ub3$, based
on Theorem \ref{generalub}, where the first three terms of $p_F$ are
evaluated exactly and the rest is approximated. All three upper bounds
approach to each other and the exact $p_F$-value, when the dependency
becomes stronger.

\begin{figure}
\begin{center}
\includegraphics[width=0.7\textwidth]{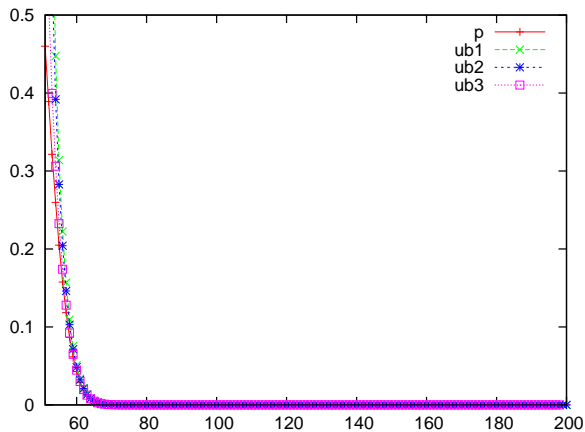}
\caption{Exact $p_F$ and three upper bounds as functions of $m(XA)$, when $m(X)=200$, $m(A)=250$, and $n=1000$. The strength of the dependency increases on the $x$-axes.}
\label{examplefig}
\end{center}
\end{figure}

Figure \ref{zoomfig} shows a magnified area from Figure
\ref{examplefig}. In this area, the dependencies are weak, the upper
bounds diverge from the exact $p_F$. The reason is that in this area
the number of approximated terms is also the largest. For example,
when $m(XA)=55$, $p_F$ contains 146 terms, and when $m(XA)=65$, it
contains $136$ terms. In these points the lift is $1.1$ and
$1.3$, respectively. The difference between $ub1$ and $ub2$ is 
marginal, but $ub3$ clearly improves $ub1$.

\begin{figure}
\begin{center}
\includegraphics[width=0.7\textwidth]{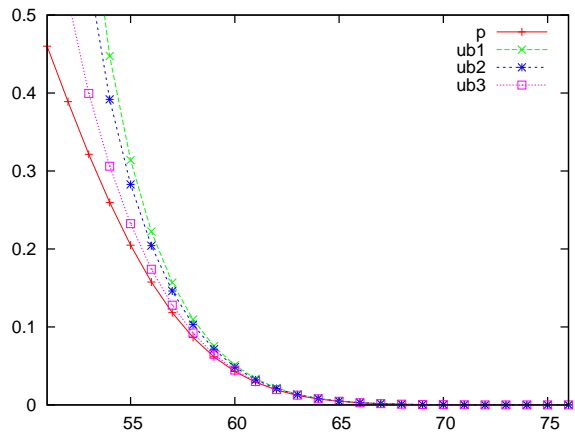}
\caption{A magnified area from Figure \ref{examplefig} showing the differences, when the dependency is weak.}
\label{zoomfig}
\end{center}
\end{figure}

Because the new upper bound gives accurate approximations for strong
dependencies, we evaluate the approximations only for the potentially 
problematic weak dependencies. As an example, we consider two data
sets, where the data size is either $n=1000$ or $n=10000$. For both
data sets, we have three test cases: 1) when $P(X)=P(A)=0.5$, 2) when
$P(X)=0.2$ and $P(A)=0.25$, and 3) when $P(X)=0.05$ and
$P(A)=0.2$. (The second case with $n=1000$ is shown in Figures
\ref{examplefig} and \ref{zoomfig}.)  For all test cases we have
calculated the exact $p_F$, three versions of the upper bound $ub1$,
$ub2$ and $ub3$, and the $p$-value achieved from the one-sided
$\chi^2$-measure. The $\chi^2$-based $p$-values were calculated with
an online Chi-Square Calculator \cite{calculator2}. The
values are reported for the cases, where $p_F\approx 0.05$,
$p_F\approx 0.01$, and $p_F\approx 0.001$. Because the data is
discrete, the exact $p_F$-values always deviate somewhat from the
reference values.

The results for the first data set ($n=1000$) are given in Table
\ref{pFapprcomp1000} and for the second data set ($n=10000$) in Table
\ref{pFapprcomp10000}. As expected, the $\chi^2$-approximation works
best, when the data size is large and the distribution is balanced
(case 1). According to the classical rule of a thumb, the
$\chi^2$-approximation can be used, when all expected counts are $\geq
5$ \cite{fisher1925}. This requirement is not satisfied in the
third case in the smaller data set. The resulting $\chi^2$-based
$p$-values are also the least accurate, but the $\chi^2$-test produced
inaccurate approximations also for the case 2, even if the smallest
expected frequency was 50. 

In the smaller data set, the $\chi^2$-approximation overperformed the
new upper bounds only in the first case, when $p_F\approx 0.05$. If we
had calculated the first four terms exactly, the resulting $ub4$ would
have already produced a better approximation. 

In the larger data set, $\chi^2$ gave more accurate results for the
first two cases, when $p_F\approx 0.05$ and for the case 1, when
$p_F\approx 0.01$. When $p_F\approx 0.001$, the new upper bounds gave
always more accurate approximations. If we had calculated the first
eight terms exactly, the resulting $ub8$ would have overperformed the
$\chi^2$-approximation in case 1 with $p_F\approx 0.01$ and case 2
with $p_F\approx 0.05$. Calculating eight exact terms is quite
reasonable compared to all 2442 terms, which have to be calculated for
the exact $p_F$ in case 1. With 15 exact terms, the approximation for
the case 1 with $p_F \approx 0.05$ would have also been more accurate
than the $\chi^2$-based approximation. However, in so large data set
(especially with an exhaustive search), a $p$-value of 0.05 (or even
0.01) is hardly significant. Therefore, we can conclude that for
practical search purposes the new upper bounds give better
approximations to the exact $p_F$ than the $\chi^2$.

\begin{table}
\caption{Comparison of the exact $p_F$ value, three upper bounds, and the 
$p$-value based on the $\chi^2$-measure, when $n=1000$. 
The best approximations are emphasized. $E(m(XA))$ is the expected 
frequency of $XA$ under the independence assumption.}
\label{pFapprcomp1000}
\begin{center}
\small{
\begin{tabular}{|l|l|l|l|l|l|}
\hline
\multicolumn{6}{|l|}{1) $n=1000$, $m(X)=500$, $m(A)=500$, $E(m(XA))=250$}\\
$m(XA)$&$p_F$&$ub1$&$ub2$&$ub3$&$p(\chi2)$\\
\hline
263&0.0569&0.0696&0.0674&0.0617&{\bf 0.050}\\
269&0.0096&0.0107&0.0105&{\bf 0.0100}&0.0081\\
275&0.00096&0.00103&0.00101&{\bf 0.00100}&0.00080\\
\hline
\multicolumn{6}{|l|}{2) $n=1000$, $m(X)=200$, $m(A)=250$, $E(m(XA))=50$}\\
$m(XA)$&$p_F$&$ub1$&$ub2$&$ub3$&$p(\chi2)$\\
\hline
60&0.0429&0.0508&0.0484&{\bf 0.0447}&0.0340\\
63&0.0123&0.0137&0.0132&{\bf 0.0125}&0.088\\
68&0.00089&0.00094&0.00092&{\bf 0.00089}&0.00050\\
\hline
\multicolumn{6}{|l|}{3) $n=1000$, $m(X)=50$, $m(A)=200$, $E(m(XA))=1$}\\
$m(XA)$&$p_F$&$ub1$&$ub2$&$ub3$&$p(\chi2)$\\
\hline
15&0.0559&0.0655&0.0605&{\bf 0.0565}&0.0349\\
17&0.0123&0.0135&0.0128&{\bf 0.0124}&0.0056\\
19&0.00194&0.00205&0.00198&{\bf 0.00194}&0.00050\\
\hline
\end{tabular}
}
\end{center}
\end{table}

\begin{table}
\caption{Comparison of the exact $p_F$ value, three upper bounds, and the 
$p$-value based on the $\chi^2$-measure, when $n=10 000$. 
The best approximations are emphasized. $E(m(XA))$ is the expected frequency 
of $XA$ under the independence assumption.}
\label{pFapprcomp10000}
\begin{center}
\small{
\begin{tabular}{|l|l|l|l|l|l|}
\hline
\multicolumn{6}{|l|}{1) $n=10000$, $m(X)=5000$, $m(A)=5000$, $E(m(XA))=2500$}\\
$m(XA)$&$p_F$&$ub1$&$ub2$&$ub3$&$p(\chi2)$\\
\hline
2541&0.0526&0.0655&0.0647&0.0621&{\bf 0.0505}\\
2559&0.0096&0.0109&0.0109&0.0106&{\bf 0.0091}\\
2578&0.00097&0.00105&0.00104&{\bf 0.00102}&0.00090\\
\hline
\multicolumn{6}{|l|}{2) $n=10000$, $m(X)=2000$, $m(A)=2500$, $E(m(XA))=500$}\\
$m(XA)$&$p_F$&$ub1$&$ub2$&$ub3$&$p(\chi2)$\\
\hline
529&0.0504&0.0623&0.0611&0.0579&{\bf 0.047}\\
541&0.0100&0.0113&0.0112&{\bf 0.0108}&0.0090\\
554&0.00109&0.00118&0.00116&{\bf 0.00114}&0.00090\\
\hline
\multicolumn{6}{|l|}{3) $n=10000$, $m(X)=500$, $m(A)=2000$, $E(m(XA))=10$}\\
$m(XA)$&$p_F$&$ub1$&$ub2$&$ub3$&$p(\chi2)$\\
\hline
115&0.0498&0.0608&0.0583&{\bf 0.0541}&0.0427\\
121&0.0105&0.0118&0.0115&{\bf 0.0109}&0.0080\\
128&0.00106&0.00114&0.00112&{\bf 0.00108}&0.00070\\
\hline
\end{tabular}
}
\end{center}
\end{table}

\section{Conclusions}
\label{secconcl}

We have introduced a family of upper bounds, which can be used to
estimate Fisher's $p_F$ accurately. Unlike the $\chi^2$-based
approximations, these upper bounds are not sensitive to the data size,
distribution, or small expected counts. In practical data mining
purposes, the simplest upper bound produces already accurate results,
but if all existing dependencies are weak and relatively
insignificant, the results can be too inaccurate. In this case, we can
use a general upper bound, whose accuracy can be adjusted freely. In
practice, it is usually sufficient to calculate only a couple of terms
from the Fisher's $p_F$ exactly. In large data sets, this is an
important concern, because the exact $p_F$ can easily require
evaluation of thousands of terms for each tested dependency. 

\bibliographystyle{alpha}

\end{document}